\documentclass{article} 
\usepackage{eqnarray}
\usepackage{graphicx}
\usepackage{amsmath}
\usepackage{amsthm}
\usepackage{amssymb}
\usepackage{setspace}

\DeclareMathOperator{\sgn}{sgn}

\newcommand{\ket}[1]{\left\vert{#1}\right\rangle}

\newtheorem{definition}{Definition}
\newtheorem{lemma}{Lemma}
\newtheorem{theorem}{Theorem}
\newtheorem{corollary}{Corollary}

\begin{document}

\title{Fourier $1$-norm and quantum speed-up}
\author{Sebasti\'an Alberto Grillo\\
Universidad Aut\'onoma de Asunci\'on \\
Asunci\'on, Paraguay\\
sgrillo@uaa.edu.py\\[2em]
\and
Franklin de Lima Marquezino\\
Universidade Federal do Rio de Janeiro \\
Rio de Janeiro, Brazil\\
franklin@cos.ufrj.br
}

\maketitle

\begin{abstract}
Understanding quantum speed-up over classical computing is fundamental for the development of efficient quantum algorithms.
In this paper, we study such problem within the framework of the Quantum Query Model, which represents the probability of output $x \in \{0,1\}^n$  as a function 
$\pi(x)$.  
We present a classical simulation for output probabilities $\pi$, whose error depends on the Fourier $1$-norm of $\pi$.
Such dependence implies upper-bounds for the quotient between the number of queries applied by an optimal classical algorithm and our quantum algorithm, respectively. These upper-bounds show a strong relation between Fourier $1$-norm and quantum parallelism. We show applications to query complexity.
\end{abstract}

\section{Introduction}

A primary motivation in quantum computing is obtaining algorithms that solve problems much faster than the best classical counterparts.  The quantum and classical decision tree models allow us to prove the existence of quantum speed-up in relation to classical query for several problems~\cite{AMBXXX,Deutsch,Reich}. Query problems can be formulated as computing Boolean functions from inputs in $\{0,1\}^{n}$, with complexity being defined as the number of queries to the input, ignoring other computations~\cite{buhrman2002complexity}.  
This implies an important simplification of the analysis in comparison to problems formulated by Turing machines, where separations between complexity classes are usually much harder to prove~\cite{aaronson2010bqp}. Several quantum algorithms can be formulated within query models \cite{Jordan}, thus this formalism is powerful enough for analyzing important algorithms, such as search algorithms~\cite{Aaronson2} or even non-query algorithms as Shor's algorithm~\cite{AARONAMB}.

A complete understanding of quantum speed-up implies determining where and how it occurs. Thus, we can study such question from two distinct approaches:  determining which \textit{functions} or  which \textit{algorithms} allow a gap between quantum and classical computing. 
The first approach is intensively used 
in quantum query complexity, where effort is mainly invested in obtaining bounds for complexity measures and checking their tightness~\cite{Aaronson2,buhrman2002complexity}. 
The second approach is commonly implemented by identifying which quantum features are hard to simulate within classical sources~\cite{abbott2010understanding}. One of the earliest attempts to explain quantum advantage is the discussion of quantum parallelism in quantum algorithms~\cite{Deutsch2}.

A well studied quantum feature is quantum entanglement~\cite{JOZSA2}, which has been identified as a necessary condition for quantum speed-up in pure-state algorithms~\cite{JOZSA}. At the same time, the study of quantum entanglement depends on whether pure or mixed quantum states are allowed~\cite{DATTA2} and the measure defined for such entanglement~\cite{JOZSA,VIDAL}. As an example of a widely applied entanglement measure, we can consider the size of partitions that describe product states in the quantum algorithm. If the size of the subsets in those partitions is upper-bounded by a constant through all the steps of the quantum algorithm, then it has an efficient classical simulation~\cite{JOZSA}. In addition, we can analyze the entanglement in a quantum state by measuring the Schmidt rank, where a polynomial upper-bound for this measure implies a polynomial classical simulation~\cite{VIDAL}. Using a model previously defined by Knill and Laflamme~\cite{KNILL}, different conditions for quantum speed-up were also identified. Such conditions are formulated on quantum correlations that are analyzed by a measure known as \textit{quantum discord}~\cite{DATTA}.
A recent proposal comes from \textit{no-go} theorems, identifying contextuality~\cite{BELL,KOCHEN}  as a necessary condition for quantum speed-up---this condition presents an inequality\- violation in contrast to the other conditions based in measures~\cite{HOWARD}. 
The identification of necessary conditions for quantum advantages
is an important issue for theoretical purposes and for the design of better quantum algorithms, specially if the conditions can be monitored in our design. 
Summarizing, a general goal in this line of research is to obtain sufficient and necessary conditions for quantum speed-up.

The present work offers a new perspective about speed-up produced by quantum algorithms in the Quantum Query Model (QQM), which is the quantum generalization for decision tree models.  First, we consider that the probability of obtaining a given output is a linear combination of orthonormal functions, where such set of functions is denoted as  \textit{Fourier basis} \cite{DEWOLF1}. 
This approach is usually referred to as \textit{analysis of Boolean functions}, and has several results in quantum query and computer science \cite{rotteler2010quantum,ozols2013quantum,montanaro2012some,o2014analysis}.  
Using such representation of the output probability, we define a classical simulation of the quantum algorithm. The idea of our simulation is implementing minor simulations for parity functions from the Fourier basis, where each simulated function appears in the Fourier decomposition of the output probability.
Similarly to related works in the context of quantum entanglement or quantum discord, in this paper we follow a strategy known as \textit{dequantization}~\cite{abbott2010understanding}, which consists in analyzing how hard is the simulation of some algorithm in relation to a given measure. We prove that the error in our simulation depends on the $L_{1}$ norm defined over the Fourier basis, where such norm is computed for the output probability. 
Thereby, a necessary property for a hard classical simulation of a given quantum algorithm, is having a large Fourier $1$-norm for its output probability functions. This necessary condition is formalized 
as an upper bound for the quotient between the number of queries of 
an optimal classical algorithm and of the simulated quantum algorithm, respectively. Notice that a well designed algorithm in the QQM setting should maximize such quotient. 

The state of any algorithm in the QQM can be described as a sum of vectors whose phases change depends on the input. The phase of each of these vectors may depend on different values from input, which shows quantum parallelism in action~\cite{GRILLO}. We show that the minimum size and the number of such vectors 
limit the value of the  Fourier $1$-norm, which allows alternative necessary conditions for quantum speed-up. The Fourier $1$-norm is maximized by the homogeneity on the size of the vectors. Which implies that simulating such balanced probabilities can be expensive by classical means. Therefore, our results give more formalism to the notion of quantum parallelism. Finally, we show applications of our results on (i) upper-bounds for randomized query, (ii) lower-bounds for exact quantum query and (iii) polynomial simulation by randomized query. 

This work is structured as follows. In Sec.~\ref{S2}, we introduce preliminary formulations and theorems about the QQM. In Sec.~\ref{S3}, we describe a classical simulation of quantum algorithms.
In Sec.~\ref{uper}, we present the upper bounds  
from our simulation. In Sec.~\ref{S5}, we present alternative applications of our results. In Sec.~\ref{S6}, we present our conclusion.

\section{Preliminary notions}
\label{S2}

The QQM~\cite{BARNUM} describes  algorithms computing functions whose domain is 
a subset of 
$\left\{0,1\right\}^{n}$. We describe the states and operations within such model, over a Hilbert space $\mathcal{H}$ with basis states $\left|i\right\rangle \left|j\right\rangle$, where $i\in\left\{ 0,1,..,n\right\}$ and $j\in\left\{ 1,..,m\right\}$, for an arbitrary $m$. The query operator is defined as
$O_{x}\left|i\right\rangle \left|j\right\rangle =\left(-1\right)^{x_{i}}\left|i\right\rangle \left|j\right\rangle$,
where $x \equiv x_0 x_1 \cdots x_n$ is the input, and $x_{0} \equiv 0$. The final state of the algorithm over input $x$ is defined as
$\left|\Psi_{x}^{f}\right\rangle =U_{t}O_{x}U_{t-1}...O_{x}U_{0}\left|\Psi\right\rangle$, where $\left\{ U_{i}\right\}$ is a set of unitary operators over $\mathcal{H}$ and $\left|\Psi\right\rangle$ is a fixed state in $\mathcal{H}$. The number of
queries or steps is defined as the times that $O_{x}$ occurs in the
algorithm.  

\begin{definition}
An indexed set of pairwise
orthogonal projectors $\left\{ P_{z}:z\in T\right\}$ is called a Complete
Set of Orthogonal Projectors (CSOP) if it satisfies
\begin{equation}
  \sum_{z\in T}P_{z}=I_{\mathcal{H}},
\end{equation}
taking $I_{\mathcal{H}}$ as the identity operator for $\mathcal{H}$.
\end{definition}

Given a CSOP defined for the algorithm, the probability of obtaining the output $z\in T$ is
$\pi_{z}\left(x\right)=\left\Vert P_{z}\left|\Psi_{x}^{f}\right\rangle
\right\Vert ^{2}$. 
We say that an algorithm computes a function $f:D\rightarrow T$ within error $\varepsilon$ if
 $\pi_{f\left(x\right)}\left(x\right) \geq 1-\varepsilon$ for all input $\ensuremath{x}\in D\subset\left\{ 0,1\right\} ^{n}$.
 
\subsection{An alternative formulation for the QQM}

In this section, we introduce notation from a previous work~\cite{GRILLO}. We define a
product of unitary operators
$\widetilde{U}_{n}=U_{n}U_{n-1}\ldots U_{0}$. We denote a CSOP $\left\{
  \bar{P}_{k}:0\leq k\leq n\right\}$, where the range of each $\bar{P}_{i}$ 
 is composed by vectors 
of the form~$\left|i\right\rangle
\left|\psi\right\rangle \in\mathcal{H}$, for $i\in\left\{ 0,1,..,n\right\}$ and any state $\left|\psi\right\rangle.$ We also introduce the notation
${\widetilde{P}_{i}^{j}=\widetilde{U}_{j}^{\dagger}\bar{P}_{i}\widetilde{U}_{j}}$. Notice that for any fixed $j$ we have that $\left\{ \widetilde{P}_{k}^{j}:0\leq
  k\leq n\right\}$ is also a CSOP. The following definition introduces an alternative representation for quantum query algorithms on the QQM.
 \begin{definition}
  \label{debs}
  Consider a set 
  $\mathbb{Z}_{n+1}=\left\{ 0,1,\ldots,n\right\}$. An
  indexed set of vectors $\left\{
    \left|\Psi\left(k\right)\right\rangle \right.$$\in
  \mathcal{H}:$$\left.k\in \mathbb{Z}_{n+1}^{t+1} \right\}$ \emph{is
    associated with} a quantum query algorithm
        if we have that
  \begin{equation}\label{ddesc}
    \left|\Psi\left(a\right)\right\rangle =\widetilde{P}_{a_{t}}^{t}\ldots\widetilde{P}_{a_{1}}^{1}\widetilde{P}_{a_{0}}^{0}\left|\Psi\right\rangle,
  \end{equation} for all $a\in \mathbb{Z}_{n+1}^{t+1}$.
\end{definition}

In Lemma~\ref{desc1}, we show that vectors associated with some algorithm represent the final state as phase flips~\cite{GRILLO}. In Sec.~4, we analyze the relation between minimum norm 
(or cardinality) 
of such vectors, and the computational gap between classical and quantum query.

\begin{lemma}
\label{desc1}
If 
the indexed set of vectors $\left\{ \left|\Psi\left(k\right)\right\rangle \in
  H_{A}:k\in \mathbb{Z}_{n+1}^{t+1} \right\}$ is associated with a quantum algorithm then
\begin{equation}
  \label{desc}
  \widetilde{U}_{t}^{\dagger}O_{x}U_{t} \ldots U_{1}O_{x}U_{0}\left|\Psi\right\rangle =\sum_{k_{t}=0}^{n} \ldots \sum_{k_{0}=0}^{n} (-1)^{\sum_{i=0}^{t}x_{k_{i}}}\left|\Psi\left(k_{0}, \ldots ,k_{t}\right)\right\rangle.
\end{equation}
\end{lemma}

\begin{proof}
Following Ref.~\cite{GRILLO}, we give a proof by induction on $t$. For $t=0$, we have that Eq. (\ref{desc}) holds,
\begin{eqnarray}
\widetilde{U}_0^\dagger O_x U_0 \ket{\Psi} &=& U_0^\dagger O_x U_0 \ket{ \Psi }\\
  &=& \sum_{k_0 = 0}^n
  (-1)^{ x_{k_{0}} }\left|\Psi\left(k_{0}\right)\right\rangle.
\end{eqnarray}
For the second part of the induction, we shall notice that the equation
\begin{equation}
\label{ooo}
   O_{x}\left|\Psi\right\rangle =\sum_{i\in\left\{ k:x_{k}=0\right\} }\bar{P}_{i}\left|\Psi\right\rangle -\sum_{i\in\left\{ k:x_{k}=1\right\} }\bar{P}_{i}\left|\Psi\right\rangle
\end{equation} implies the equation \begin{equation}
  \label{eqx}
  \widetilde{U}_{j}^{\dagger}O_{x}\widetilde{U}_{j}\left|\Psi\right\rangle =\sum_{i\in\left\{ k:x_{k}=0\right\} }\widetilde{U}_{j}^{\dagger}\bar{P}_{i}\widetilde{U}_{j}\left|\Psi\right\rangle -\sum_{i\in\left\{ k:x_{k}=1\right\} }\widetilde{U}_{j}^{\dagger}\bar{P}_{i}\widetilde{U}_{j}\left|\Psi\right\rangle.
\end{equation}
Suppose that Eq. (\ref{desc}) holds for some $t$, then applying Eq. (\ref{eqx}) we obtain 
\begin{multline*}\widetilde{U}_{t+1}^\dagger O_x U_t\ldots U_1 O_x U_0 \ket{\Psi} =\\ \sum_{k_{t}=0}^{n} \ldots\sum_{k_{0}=0}^{n}(-1)^{\sum_{i=0}^{t} x_{k_{i}} } \sum_{k_{t+1}=0}^{n}(-1)^{x_{k_{t+1}}}\widetilde{P}_{k_{t+1}}^{t+1}\left|\Psi\left(k_{0},\ldots ,k_{t}\right)\right\rangle=\\ \sum_{k_{t+1}=0}^{n} \ldots\sum_{k_{0}=0}^{n}(-1)^{\sum_{i=0}^{t+1} x_{k_{i}} }\widetilde{P}_{k_{t+1}}^{t+1}\left|\Psi\left(k_{0},\ldots,k_{t}\right)\right\rangle =\\ \sum_{k_{t+1}=0}^{n} \ldots\sum_{k_{0}=0}^{n}(-1)^{\sum_{i=0}^{t+1} x_{k_{i}} }\left|\Psi\left(k_{0},\ldots ,k_{t+1}\right)\right\rangle.\end{multline*}
\end{proof}

The previous theorem shows that a quantum state depends on several components whose phases change  independently on input~$x$. Notice that the phase $(-1)^{\sum_{i=0}^{t}x_{k_{i}}}$ of each component $\left|\Psi\left(k_{0}, \ldots ,k_{t}\right)\right\rangle$ is a Walsh function. Then, each of the components depends on $t$ values from input, which at first sight is not impressive, considering that deterministic classical algorithms  compute any function that depends on $t$ values using $t$ queries. However, all components together depend on the size $n$ of input. Thus, we have the possibility of computing on $n$ variables  using just $t$ queries, which gives us another intuition about the computational speed-up by quantum means. Therefore, this formulation presents quantum parallelism more explicitly than a sequence of unitary operators.  

\section{A classical simulation for quantum query algorithms and polynomials} \label{S3}

In this section, we introduce our simulation of quantum query algorithms by classical algorithms. However, our simulation can also be extended to polynomials. This simulation is defined over the output probability $\pi_{z}\left(x\right)$ of the quantum algorithm.

We consider the Fourier basis for the vector space of all functions $f:\left \{ 0,1 \right \}^{n}\rightarrow \mathbb{R}$~\cite{DEWOLF1} given by the functions
$$\chi_{b}:\left\{ 0,1\right\} ^{n}\rightarrow\left\{ 1,-1\right\},$$ such that $\chi_{b}(x)=\left ( -1 \right )^{b\cdot x}$ for $b\in \left \{ 0,1 \right \}^{n}$ and $b\cdot x=\sum_{i} b_{i}x_{i}$. This family contains a constant function that we denote as $\chi _{0}=1$.  
Therefore, any function $f:\left \{ 0,1 \right \}^{n}\rightarrow \mathbb{R}$ can be represented as a  linear combination 
\begin{equation}\label{sum}
f=\underset{b\in \left \{ 0,1 \right \}^{n}}{\sum}\alpha_{b}\chi _{b},
\end{equation} and we denote the Fourier $1$-norm of $f$
 as \begin{equation}L\left(f\right)=\underset{b\in \left \{ 0,1 \right \}^{n}}{\sum}\left|\alpha_{b}\right|.\end{equation} Another measure is the degree of $f$, which is defined as 
 \begin{equation} \deg\left ( f \right )=\max_{\left | b \right |} \left \{ b:\alpha _{b}\neq 0 \right \}, 
 \end{equation} where $|b|$ denotes the number of ones in $b$.
 
\begin{figure}\centering\includegraphics[width=0.9\textwidth]{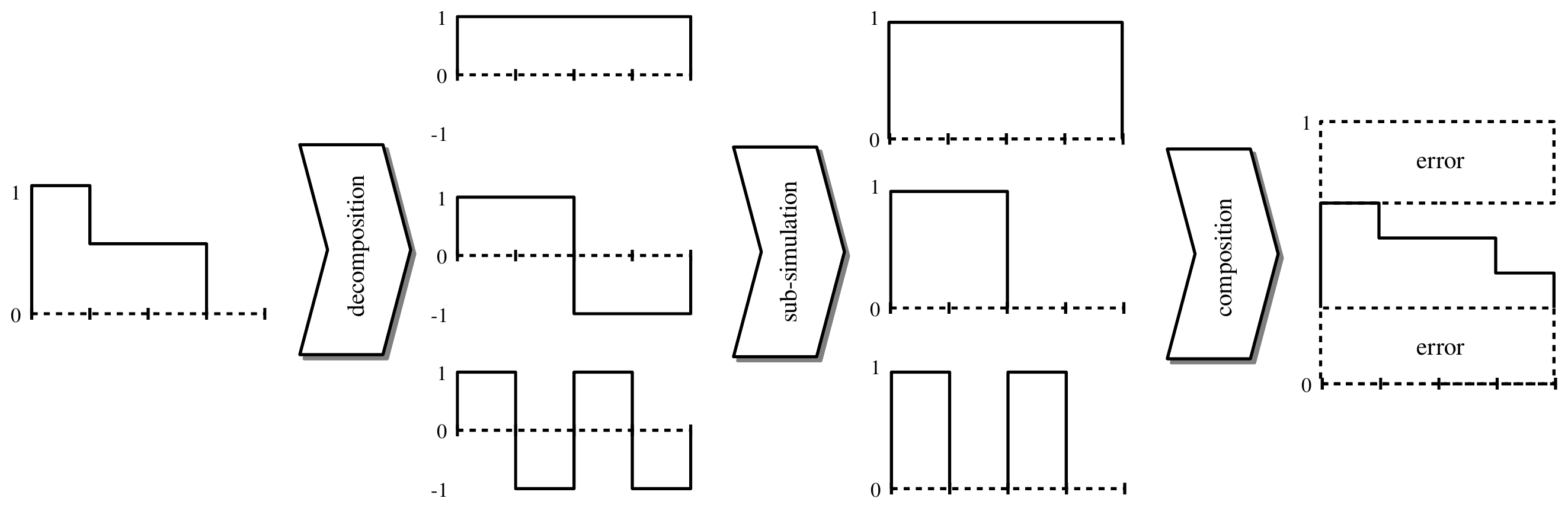}\caption{The simulation produces a contracted version of the original output probability. The new output probability can be represented as a linear transformation applied over the original output probability.\label{f111}}\end{figure}

Figure \ref{f111} presents the intuition behind our simulation. At the right, we have $\pi_{1}\left(x\right)$, the probability of obtaining output~$1$ by 
a quantum query algorithm on input $x$.
Such function is decomposed into a linear combination of functions~$\chi _{b}$, following Eq.~(\ref{sum}). The sub-simulations imply emulating each function $\chi _{b}$ by a classical algorithm that outputs $1$ with probability $\widehat{\pi}_{1}^{b}\left(x\right),$ where: (i) $\chi _{b}(x)=1$ implies that $\widehat{\pi}_{1}^{b}\left(x\right)=1$; and, (ii) $\chi _{b}(x)=-1$ implies that $\widehat{\pi}_{1}^{b}\left(x\right)=0$. Notice that each $\widehat{\pi}_{1}^{b}\left(x\right)$ is a probability and can not have negative values as functions $\chi _{b}$. The composition step is assigning appropriate probabilities to each output $\widehat{\pi}_{1}^{b}\left(x\right),$ such that the sum produces an output probability whose shape resembles $\pi_{1}\left(x\right)$. As each $\widehat{\pi}_{1}^{b}\left(x\right)$ is similar but different in relation to $\chi _{b}$, this procedure accumulates an important error. The proof of the following theorem shows the details.
\begin{theorem}
\label{1teo1}
Let $\mathcal{A}$ be a quantum algorithm that computes $f:S\rightarrow\left\{ 0,1\right\} $
for $S\subset\left\{ 0,1\right\} ^{n}$, within error $\varepsilon$ and $t$ queries.
  Then, there is a classical algorithm which computes $f$ within error
$$\tilde{\varepsilon}=\frac{\varepsilon+L\left(\pi_{1}\right)}{1+2L\left(\pi_{1}\right)}$$ 
and $2t$ queries.
\end{theorem}
\begin{proof}
If a quantum algorithm applies $t$ queries, then $\deg\left ( \pi_{z} \right ) \leq 2t$ for every output $z$ \cite{beals2001quantum}.
Let $\mathcal{D}\left(b\right)$ be the deterministic classic algorithm which
outputs %$1$ with probability given by
\begin{equation}\label{2teo5}\widehat{\pi}_{1}^{b}\left(x\right)=\frac{1}{2}+\sgn\left(\alpha_{b}\right)\left(\frac{\chi _{b}(x)}{2}\right),\end{equation} for input $x$, where $\sgn$ is the sign function and $|b| \leq 2t$. We consider a randomized algorithm $\mathcal{R}$
which simply selects either: (i) an algorithm  $\mathcal{D}\left(b\right)$, with
probability $\frac{2\left|\alpha_{b}\right|}{1+2L\left(\pi_{1}\right)};$
 or, (ii) an algorithm that outputs 0 for any $x$, with probability $\frac{1}{1+2L\left(\pi_{1}\right)}$. Notice that algorithm $\mathcal{R}$ is the composition of sub-simulations, as we represent in Figure \ref{f111}.
Since we denote by $\widehat{\pi}_{1}\left(x\right)$ the probability
of obtaining output $1$ given $x$ with $\mathcal{R}$, by Eq.~(\ref{2teo5}) we have
\begin{align}
\widehat{\pi}_{1}\left(x\right)&=\frac{\underset{b}{\sum}2\left|\alpha_{b}\right|\widehat{\pi}_{1}^{b}\left(x\right)}{1+2L\left(\pi_{1}\right)}\\
&=\frac{\underset{b}{\sum}\left|\alpha_{b}\right|+\underset{b}{\sum}\alpha_{b}\chi_{b}(x)}{1+2L\left(\pi_{1}\right)}.
\end{align} The algorithm $\mathcal{R}$ applies no more than $2t$ queries, 
since $\mathcal{D}\left(b\right)$ applies no more than $2t$ queries for each $|b| \leq 2t$.

Now, we must prove an upper bound for the error in the simulation. We divide such proof in two cases, when $f\left(x\right)=1$ and $f\left(x\right)=0$. If $f\left(x\right)=1$, then $\varepsilon\geq1-\pi_{1}\left(x\right)=1-\underset{b}{\sum}\alpha_{b}\chi_{b}(x)$. This implies that
\begin{align} 
1-\widehat{\pi}_{1}\left(x\right) &= 1-\frac{\left(L\left(\pi_{1}\right)+\underset{b}{\sum}\alpha_{b}\chi_{b}(x)\right)}{1+2L\left(\pi_{1}\right)}\\
&= \frac{1+L\left(\pi_{1}\right)-\underset{b}{\sum}\alpha_{b}\chi_{b}(x)}{1+2L\left(\pi_{1}\right)}\\
&\leq \tilde{\varepsilon}.
\end{align}
Analogously, if $f\left(x\right)=0$, then $\varepsilon\geq\pi_{1}\left(x\right)=\underset{b}{\sum}\alpha_{b}\chi_{b}(x)$
and this implies that \begin{equation}\widehat{\pi}_{1}\left(x\right)\leq\frac{\varepsilon+L\left(\pi_{1}\right)}{1+2L\left(\pi_{1}\right)}=\tilde{\varepsilon}.\end{equation}
\end{proof}

We described a classical simulation that imitates the output probability of a given quantum algorithm, but within a big error. Thus, the next theorem just gives a reduction of such error using probabilistic amplification.
\begin{theorem}
\label{1theo2}
Let $\mathcal{A}$ be a quantum algorithm that computes $f:S\rightarrow\left\{0,1\right\}$ for $S\subset\left\{ 0,1\right\} ^{n}$, with error $\varepsilon$ and $t$ queries. 
Then, there is a classical algorithm which computes $f$ within error
$\exp\left({-\frac{j}{2\left(1-\tilde{\varepsilon}\right)}\left(\frac{1}{2}-\tilde{\varepsilon}\right)^{2}}\right)$,
where $\tilde{\varepsilon}=\frac{\varepsilon+L\left(\pi_{1}\right)}{1+2L\left(\pi_{1}\right)}$
and using $2jt$ queries.
\end{theorem}
\begin{proof}
We use a corollary of Chernoff bound~\cite{angluin1979fast}. For $j,p,\beta$ such
that $0\leq p\leq1$, $0\leq\beta\leq1$ and $0\leq j$, we have
\begin{equation}
\label{1eq1}
\sum_{i=0}^{m}
\binom{j}{i}
p^{i}\left(1-p\right)^{j-i}\leq \exp{\left(-\beta^{2}jp/2\right)},
\end{equation}
where $m = \left\lfloor \left(1-\beta\right)jp\right\rfloor$.

We define an algorithm $\mathcal{\widehat{R}}$ using the classical
algorithm $\mathcal{R}$ within error $\tilde{\varepsilon}$ from Theorem~\ref{1teo1}. Algorithm $\mathcal{\widehat{R}}$ consists in applying probability amplification
on $\mathcal{R}$, that is, executing algorithm $\mathcal{R}$ $j$ times and then selecting the most
frequent result. Define $X$ as the random variable that represents
the number of correct answers. Taking $\beta=1-\frac{1}{2\left(1-\tilde{\varepsilon}\right)}$
and $p=\left(1-\tilde{\varepsilon}\right)$ in Eq.~(\ref{1eq1}), then the
error in $\mathcal{\widehat{R}}$ is upper-bounded by
\begin{equation}
\label{1eq11}
\mathbb{P}\left[X\leq\left\lfloor \frac{j}{2}\right\rfloor \right]\leq \exp\left({-\frac{j}{2\left(1-\tilde{\varepsilon}\right)}\left(\frac{1}{2}-\tilde{\varepsilon}\right)^{2}}\right).
\end{equation}
\end{proof}

\subsection{Polynomial simulation}

The same technique can be applied for simulating a polynomial $p\left(x\right)$ approximating a function, instead of simulating the output probabilities of a given quantum algorithm. In this sense, Theorems \ref{1teo1} and \ref{1theo2} can be generalized. In order to formulate the corresponding theorems, we consider the usual notion of polynomial approximation: 

\begin{definition} A polynomial $p:\mathbb{R}^n\rightarrow \mathbb{R}$ $\varepsilon$-approximates a function $f:S\rightarrow\left\{ 0,1\right\} $
for $S\subset\left\{ 0,1\right\} ^{n}$, if $\left | p\left ( x \right )-f\left ( x \right ) \right |\leq \varepsilon$ for all $x \in \left \{ 0,1 \right \}^n$.
\end{definition}

\begin{theorem}\label{pol1}
Let $p:\mathbb{R}^n\rightarrow \mathbb{R}$ be a polynomial that $\varepsilon$-approximates $f:S\rightarrow\left\{ 0,1\right\} $
for $S\subset\left\{ 0,1\right\} ^{n}$. If $p$ has a degree equal or less than $2t$, then there is a classical algorithm which computes $f$ within error $$\tilde{\varepsilon}=\frac{\varepsilon+L\left(p\right)}{1+2L\left(p\right)}$$ and $2t$ queries.
\end{theorem}\begin{proof}
The proof is similar to the proof of Theorem \ref{1teo1}. 
\end{proof} 

We similarly introduce the corresponding reduction error theorem.

\begin{theorem}
\label{pol2}
Let $p$ be a polynomial that $\varepsilon$-approximates $f:S\rightarrow\left\{ 0,1\right\} $
for $S\subset\left\{ 0,1\right\} ^{n}$.
If $p$ has a degree equal or less than $2t$, then there is a classical algorithm which computes $f$ within error
$\exp\left({-\frac{j}{2\left(1-\tilde{\varepsilon}\right)}\left(\frac{1}{2}-\tilde{\varepsilon}\right)^{2}}\right)$,
where $\tilde{\varepsilon}=\frac{\varepsilon+L\left(p\right)}{1+2L\left(p\right)}$
and using $2jt$ queries.
\end{theorem}\begin{proof}
The proof is similar to the proof of Theorem \ref{1theo2}, but reducing error in Theorem \ref{pol1} instead Theorem \ref{1teo1}. 
\end{proof}

\section{Upper bounds for quantum speed-up}
\label{uper}

In this section, we describe conditions which can slow down our simulation. Quantum speed-up only occurs when no classical simulation is efficient enough, thus any condition that makes difficult any classical simulation is a necessary condition for this computational gain. In this sense, we measure the quantum speed-up for a given quantum algorithm by the quotient $\frac{R}{t}$, where (i) such quantum algorithm applies $t$ queries, and (ii) an optimal classical algorithm executes the same computational task in $R$ queries. This quotient can be interpreted as how much faster is a quantum algorithm in relation to the best classical algorithm.

The following theorem, which upper-bounds quantum speed-up using Fourier $1$-norm, is the core of our results. It basically shows how high values for Fourier $1$-norm are related to the speed quotient that we denoted.

\begin{theorem}\label{1theo3}
Consider $D\subset \left\{ 0,1\right\}^{n}$ and a function $f:D\rightarrow\left\{ 0,1\right\} $
 that is computed within error $\varepsilon>0$ and $t$ queries,
by a quantum query algorithm. If we define
\begin{equation}
F_{\varepsilon}\left(l\right)=\left\lceil \frac{-16\ln\left(\varepsilon\right)\left(1+l\right)\left(1+l-\varepsilon\right)}{\left(1-2\varepsilon\right)^{2}}\right\rceil, 
\end{equation} then
\begin{equation}\label{1eq66}
\frac{R_{\varepsilon}\left(f\right)}{t}\leq F_{\varepsilon}\left(L\left(\pi_{1}\right)\right),
\end{equation}
where (i) $R_{\varepsilon}\left(f\right)$ denotes the minimum number of queries that are necessary for computing $f$ within error $\varepsilon$ by a randomized decision tree (See~\cite{buhrman2002complexity} for a detailed definition.) and (ii) $\pi_{1}(x)$ is the probability of the quantum algorithm returning output 1 for a given input $x$.
\end{theorem}
\begin{proof}
Suppose that we simulate the quantum algorithm using the randomized
algorithm of Theorem \ref{1theo2} and  promising an error that does not exceed $\varepsilon$ for $f$.
Thereby, from Eq.~(\ref{1eq11}), we have 
\begin{equation}
\label{1eq6}
\varepsilon = \exp\left( -\frac{j}{2\left(1-\tilde{\varepsilon}\right)}\left(\frac{1}{2}-\tilde{\varepsilon}\right)^{2} \right).
\end{equation}
As $\frac{R_{\varepsilon}\left(f\right)}{t}\leq\left\lceil 2j\right\rceil $,
if we obtain $j$ from Eq.~(\ref{1eq6}) we have Eq.~(\ref{1eq66}).
\end{proof}

Last theorem has consequences in exact quantum complexity, as we find in the following corollary:

\begin{corollary}\label{cor}Consider a total function $f:D\rightarrow\left\{ 0,1\right\}$, then 
\begin{equation}\label{1eq7}
\frac{R_{\varepsilon}\left(f\right)}{Q_{E}\left(f\right)}\leq F_{\varepsilon}\left(L\left(f\right)\right),
\end{equation} where $Q_{E}\left(f\right)$ denotes the number of queries applied by an exact quantum query algorithm computing $f$.\end{corollary}\begin{proof}
If a quantum query algorithm is exact, optimal and computes a total function then $t=Q_{E}\left(f\right)$ and $\pi_{1}=f$.\end{proof}

Theorem \ref{1theo3} can also be formulated for approximate polynomials, as follows:

\begin{theorem}\label{pol3}
Consider $D\subset \left\{ 0,1\right\}^{n}$ and a function $f:D\rightarrow\left\{ 0,1\right\}$
 that is $\varepsilon$-approximated by a polynomial $p:\mathbb{R}^n\rightarrow \mathbb{R}$. If $\deg\left(p\right)\leq 2t,$ then
\begin{equation}\label{1eq8}
\frac{R_{\varepsilon}\left(f\right)}{2t}\leq F_{\varepsilon}\left(L\left(p\right)\right).
\end{equation}\end{theorem} \begin{proof}
Similar proof as for Theorem \ref{1theo3}, but applying Theorem \ref{pol2}.
\end{proof}

We may expect from Fourier $1$-norm that low values must imply problems that are easily simulated by classical means. Theorem~\ref{1theo3} guarantees that low values of the Fourier $1$-norm in relation to $t$ imply such efficient classical simulation. 

Notice that Fourier $1$-norm is defined on the output probability. Then, an explicit expression for the Fourier $1$-norm as a function of the algorithm itself may be useful.
Let $k,h$ be vectors in $\mathbb{Z}^{t}_{n+1}$ and $|b|\leq 2t$. We denote $\left(k,h\right)\sim b$,
if 
\[(-1)^{\underset{i}{\sum}x_{k_{i}}+\underset{i}{\sum}x_{h_{i}}}=\chi_{b}(x).\] 
Thus, for a $t$-query algorithm, we have the expression
\begin{equation}\label{lup1}
L\left(\pi_{1}\right)=\underset{|b|\leq 2t}{\sum}\left|\underset{\begin{array}{c}
\left(k,h\right)\sim b
\end{array}}{\sum}\left\langle \Psi\left(k\right)\right|P_{1}\left|\Psi\left(h\right)\right\rangle \right|,
\end{equation} by applying Lemma~\ref{desc1}.
Considering that each pair $\left(k,h\right)$ is related to a unique $b$, we can obtain the following upper bound for $L\left(\pi_{1}\right)$,
\begin{equation}\label{lup2}
\widetilde{L}\left(\pi_{1}\right)=\underset{k}{\sum}\underset{h}{\sum}\left|\left\langle \Psi\left(k\right)\right|P_{1}\left|\Psi\left(h\right)\right\rangle \right|.
\end{equation}

These expressions are based on the state decomposition given by Definition~\ref{debs}. Lemma~\ref{desc1} implies that each quantum algorithm has its own state decomposition, thus next theorem relates metrics on such set of vectors with the gap between quantum and classical query.
\begin{theorem}
Using the same hypothesis of Theorem~\ref{1theo3}, denoting 
$\#S$
as the cardinality of set $S$ and defining $d=\#\left\{ k:\left|\Psi\left(k\right)\right\rangle \neq0\right\}
 $, we have
\begin{equation}\label{1eq662}
\frac{R_{\varepsilon}\left(f\right)}{t}\leq F_{\varepsilon}\left(\widetilde{L}\left(\pi_{1}\right)\right),
\end{equation} 
\begin{equation} \label{up1}
\frac{R_{\varepsilon}\left(f\right)}{t}\leq F_{\varepsilon}\left(\left(\underset{k}{\sum}\left\Vert \left|\Psi\left(k\right)\right\rangle \right\Vert \right)^{2}\right),
\end{equation}
\begin{equation} \label{up2}
\frac{R_{\varepsilon}\left(f\right)}{t}\leq F_{\varepsilon}\left(d\right),
\end{equation} 
and
\begin{equation} \label{up3}
\frac{R_{\varepsilon}\left(f\right)}{t}\leq F_{\varepsilon}\left(\frac{1}{\underset{k}{\min}\,\left\langle \Psi\left(k\right)\right|\left.\Psi\left(k\right)\right\rangle }\right).
\end{equation}

\end{theorem}
\begin{proof}
As $F_{\varepsilon}$ is an increasing function, Eq.~(\ref{1eq662}) follows directly from Eq.~(\ref{lup2}) and Theorem~\ref{1theo3}. Eq.~(\ref{up1}) is also derived from Eq.~(\ref{lup2}) by observing that 
\begin{equation}
\label{1theo51}
\left|\left\langle \Psi\left(k\right)\right|P_{z}\left|\Psi\left(h\right)\right\rangle \right|\leq\left\Vert \left|\Psi\left(k\right)\right\rangle \right\Vert \left\Vert \left|\Psi\left(h\right)\right\rangle \right\Vert,
\end{equation} which gives 
\begin{equation}
\label{1theo51b}
L\left(\pi_{1}\right)\leq\left(\underset{k}{\sum}\left\Vert \left|\Psi\left(k\right)\right\rangle \right\Vert \right)^{2}.
\end{equation}
Applying Lemma~\ref{desc1}, we obtain 
\begin{equation}
\label{1theo52}
\left\langle \Psi\right|\left.\Psi\right\rangle  =\underset{k,h}{\sum}\left\langle \Psi\left(k\right)\right|\left.\Psi\left(h\right)\right\rangle=1.
\end{equation} 
Then, using it with $\underset{k}{\sum}\left\Vert \left|\Psi\left(k\right)\right\rangle \right\Vert \leq\sqrt{d}\underset{k}{\sum}\left\langle \Psi\left(k\right)\right|\left.\Psi\left(k\right)\right\rangle $ and Eq.~(\ref{1theo51b}), we have
\begin{equation}
L\left(\pi_{1}\right)\leq d.
\end{equation} 
Finally, Eq.~\eqref{up3} follows from
$d\left(\underset{k}{\min}\,\left\langle \Psi\left(k\right)\right|\left.\Psi\left(k\right)\right\rangle \right) \leq 1.$

\end{proof}

\section{Alternative applications}\label{S5}

Our results from Sec.~\ref{uper} have a main theoretical motivation, which is showing the relation between quantum speed-up and quantum parallelism. Furthermore, the theorems are interesting for related subjects that we discuss below.

\subsection{Upper bounds for randomized complexity} \label{example}

Theorems~\ref{1theo3} and \ref{pol3} may be applied for finding upper bounds on $R_{\varepsilon}$. For example, consider Deutsch-Jozsa
algorithm, thereby we have the output probability
$$\pi_{1}\left(x\right)=\frac{1}{n^{2}}\left(n-2\left|x\right|\right)^{2}$$ 
for inputs of size~$n$.
We obtain the terms $\left\{ \alpha_{b}\right\}$ by applying 
the pairwise orthogonality between functions $\chi_{b}$.
The algorithm works by applying just one query. Thus, 
from the fact that $\deg\left ( \pi_{1}\left ( x \right )  \right ) \leq 2$~\cite{beals2001quantum}, we have that if $|b|>2 $ then $\alpha_{b}=0$. This leaves us with three cases to analyze. First, if $|b|=0$, then $\alpha_{b}=\frac{1}{n}$, notice that there is just one index $b$ satisfying $|b|=0$. Second, if $|b|=1$, then $\alpha_{b}=0$.  
Third, there are $\frac{n\left(n-1\right)}{2}$ indices $b$ such that $|b|=2$, 
in this case $\alpha_{b}=\frac{2}{n^{2}}$.
Therefore, we have that $\sum\left|\alpha_{b}\right|=1$, which implies 
$$R_{\varepsilon}\leq\left\lceil \frac{-16\ln\left(\varepsilon\right)\left(2-\varepsilon\right)}{\left(1-2\varepsilon\right)^{2}}\right\rceil$$ by Eq.~\eqref{1eq66}. This is not quite tight numerically because a classical decision tree applies $2$ queries in order to solve Deutsch-Jozsa problem within error $\frac{1}{3}$. However, this is asymptotically tight and proves
that Deutsch-Jozsa algorithm can be simulated classically using a constant
number of queries and fixed error.

\subsection{Lower bounds for exact quantum complexity} 
\label{example} 

Corollary \ref{cor} can be applied for finding lower bounds on $Q_{E}$. For example, consider the total function $AND_n:\big\{0,1\big\}^{n}\rightarrow \{0,1\big\}$ where $AND_n\left(x\right)=1$ if and only if $x_{i}=1$ for all $i$. We denote \textit{weight} of input $x$ as the number of ones in $x$. A randomized decision tree computing $AND_n$ must discriminate the input with weight $n$ from the set of inputs with weight $n-1$. Suppose that some randomized decision tree computes $AND_n$ with less than $\frac{n}{3}$ queries, then such randomized tree is a probabilistic distribution over a set of deterministic decision trees querying less than $\frac{n}{3}$ values in $x$. Then, in order to discriminate an input of weight $n$ from the set of inputs with weight $n-1$, the randomized tree will find $0$ for some $x_i$ with expectation less than $\frac{1}{3}$. In this sense, $R_{\frac{1}{3}}\left(AND_n\right) \geq \frac{n}{3}-1$.

Considering that $L\left(AND_n\right)=1$, we have $Q_{E}\left(AND_n\right)\in \Omega \left ( n \right )$ by Eq.~\eqref{1eq7}, which it is asymptotically tight \cite{ambainis2015exact}.

\subsection{Polynomial approximation by quantum algorithms}

There is an equivalence between 1-query algorithms and degree-2 polynomials. That is, a partial boolean function $f$ can be approximated by a polynomial for some error bounded by $\varepsilon > \frac{1}{2}$ if and only if $f$ can be computed by a quantum algorithm with error bounded by $\varepsilon' > \frac{1}{2}$ and a single query. However, the problem of transforming polynomials of higher degree to quantum algorithms still needs more results \cite{aaronson2015polynomials}. Theorem \ref{pol1} implies that $t$-query algorithms compute any function approximated by degree-$t$ polynomials with Fourier $1$-norm bounded by a constant. Then, the high degree problem is reduced to finding algorithms for polynomials with a high Fourier $1$-norm.

\section{Conclusion}
\label{S6}
In the present work we identified a necessary property for a hard classical simulation of quantum query algorithms, namely a high Fourier $1$-norm defined over the output probability. A remarkable feature about Fourier $1$-norm is that it depends on both evolution and measurement steps. Properties like quantum entanglement are defined just on the quantum states, which implies that a poor measurement step can cancel advantages obtained in the evolution stage, where we assume that such evolution stage was hard to simulate. Nevertheless, the accuracy of Fourier $1$-norm for approximating quantum gain depends on a simulation, whose relation with the most efficient classical 
simulation is unknown.

We also formalized the advantage given by quantum algorithms, as the quotient between the classical and quantum complexities for a given task. We have that such quotient is upper-bounded by an expression which depends quadratically on the Fourier $1$-norm. Thus, a large factor produced between quantum and classical algorithms implies a 
large Fourier $1$-norm. 
Our result suggests the following intuitions: 
\begin{enumerate}
\item Output probabilities with 
large Fourier $1$-norms 
imply that such output probability can be represented by a function whose shape is much different from any function in the Fourier basis---functions that can be efficiently simulated by classical means. 
\item Output probabilities with high  Fourier $1$-norms imply that many functions from Fourier basis 
are acting simultaneously. That strongly suggests quantum parallelism.
\end{enumerate}

We  can also link Fourier $1$-norm to quantum parallelism 
as follows. A quantum query algorithm can be viewed as a state decomposition by Lemma~\ref{desc1}, which is denoted as a set of vectors associated to the algorithm. This formulation emphasizes the presence of quantum parallelism,  because each combination of vectors in the decomposition represents a function in the Fourier basis, where such functions are added producing an output probability function. The Fourier $1$-norm is related to this decomposition. Since a high Fourier $1$-norm implies: (a) a big number of non-zero vectors in such decomposition, i.e., high values for $\#\left\{k:\left|\Psi\left(k\right)\right\rangle \neq0\right\}$; and, (b) minimum product values that are not too big for such vectors, i.e., low values for $\left(\underset{k}{\min}\,\left\langle \Psi\left(k\right)\right|\left.\Psi\left(k\right)\right\rangle \right);$ then (a) and (b) are also necessary conditions for a hard classical simulation. Both measures can be linked to quantum parallelism by the following intuition. If $\#\left\{ k:\left|\Psi\left(k\right)\right\rangle \neq0\right\}$ is low, then there are less combinations of vectors adding functions on the output probability function. Larger values for $\underset{k}{\min}\,\left\langle \Psi\left(k\right)\right|\left.\Psi\left(k\right)\right\rangle$ implies lower values for $\#\left\{ k:\left|\Psi\left(k\right)\right\rangle \neq0\right\}$. However, it also implies that the output probability function has a shape closer to functions in the Fourier basis, hence such output probability has a cheap classical simulation.

Finally, the present work leaves the following open problems:

\begin{itemize}
\item Finding degree-2 polynomials is an alternative strategy for obtaining 1-query quantum algorithms \cite{aaronson2015polynomials}. Thus, developing a method for obtaining high $1$-norm polynomials of degree 2 and bounded in $\left\{0,1\right\}^{n}$ would help to find algorithms offering a potential advantage over classical algorithms. 
\item A high Fourier $1$-norm implies a necessary condition for quantum query speed-up. Can we obtain a necessary and sufficient condition by adding another property?
\end{itemize}

\section*{Acknowledgements}

This work received financial support from CAPES, FAPERJ, and CNPq. The authors thank R. Portugal, S. Collier, J. Szwarcfiter, and E. Galv\~{a}o for useful discussions and suggestions. This work was initiated while SAG was at the Federal University of Rio de Janeiro, Brazil.

\bibliographystyle{plain}
\bibliography{sample-bibliography}

\end{document}